\documentclass{article}
\usepackage[nohyperref, accepted]{icml2011}

\usepackage{color}
\usepackage{url}
\usepackage{verbatim} 
\usepackage{subfigure} 
\usepackage{multirow}
\usepackage{times}
\usepackage{hyperref}
\usepackage{algorithm}
\usepackage{algorithmic}
\usepackage{xspace}
\usepackage{graphicx}
\usepackage{epsfig}
\usepackage{amsmath}
\usepackage{amsthm}
\usepackage{amssymb}
\usepackage{natbib}
\usepackage{ctable}

\newcommand{\hide}[1]{}
\newcommand{\xhdr}[1]{\vspace{1.7mm}\noindent{{\bf #1.}}}
\newcommand{\casc}{{\mathbf{t}}}
\newcommand{\alphs}{{\mathbf{A}}}

\newtheorem{theorem}{Theorem}
\newtheorem{corollary}[theorem]{Corollary}

\newcommand{\netinf}{{\textsc{Net\-Inf}}\xspace}
\newcommand{\netrate}{{\textsc{Net\-Rate}}\xspace}
\newcommand{\connie}{{\textsc{Co\-nNIe}}\xspace}
\newcommand{\expo}{{\textsc{Exp}}\xspace}
\newcommand{\pow}{{\textsc{Pow}}\xspace}
\newcommand{\ray}{{\textsc{Ray}}\xspace}
\newcommand{\eg}{\emph{e.g.}}
\newcommand{\ie}{\emph{i.e.}}
\newcommand{\unobs}{{\infty}}

\icmltitlerunning{Uncovering the Temporal Dynamics of Diffusion Networks}

\begin{document}

\twocolumn[

\icmltitle{Uncovering the Temporal Dynamics of Diffusion Networks}

\icmlauthor{Manuel Gomez-Rodriguez$^{1,2}$}{manuelgr@stanford.edu}

\icmlauthor{David Balduzzi$^1$}{balduzzi@tuebingen.mpg.de}

\icmlauthor{Bernhard Sch\"{o}lkopf$^1$}{bs@tuebingen.mpg.de}

\icmladdress{$^1$MPI for Intelligent Systems and $^2$Stanford University} 

\icmlkeywords{inferring networks, diffusion}

\vskip 0.3in
]

\begin{abstract}
Time plays an essential role in the diffusion of information, influence and disease over networks. In many cases we only observe when a node copies information, makes a decision or becomes infected -- but the connectivity, transmission rates between nodes and transmission sources are unknown. Inferring the underlying dynamics is of outstanding interest since it enables forecasting, influencing and retarding infections, broadly construed.
To this end, we model diffusion processes as discrete networks of continuous temporal processes occurring at different rates. Given cascade data -- observed infection times of nodes -- we infer the edges of the global diffusion network and estimate the transmission rates of each edge that best explain the observed data. 
The optimization problem is convex. The model naturally (without heuristics) imposes sparse solutions and requires no parameter tuning. The problem decouples into a collection of independent smaller problems, thus scaling easily to networks on the order of hundreds of thousands of nodes. Experiments on real and synthetic data show that our algorithm both recovers the edges of diffusion networks and accurately estimates their transmission rates from cascade data.

\end{abstract}

\section{Introduction}
\label{sec:intro}
Diffusion and propagation processes have received increasing attention in a broad range of domains: information propagation~\citep{adar05epidemics, manuel10netinf, meyers10netinf}, social networks~\citep{kempe03maximizing, lappas2010finding}, viral marketing~\citep{dodds07influentials} and epidemiology~\citep{wallinga04epidemic}. 

Observing a diffusion process often reduces to noting when nodes (people, blogs, etc.) reproduce a piece of information, get infected by a virus, or buy a product. Epi\-de\-mio\-lo\-gists can observe when a person becomes ill but they cannot tell who infected her or how many exposures and how much time was necessary for the infection to take hold. In information propagation, we observe when a blog mentions a piece of information. However if, as is often the case, the blogger does not link to her source, we do not know where she acquired the information or how long it took her to post it. Finally, viral marketers can track when customers buy products or subscribe to services, but typically cannot observe who influenced customers' decisions, how long they took to make up their minds, or when they passed recommendations on to other customers. In all these scenarios, we observe \emph{where and when} but \emph{not how or why} information (be it in the form of a virus, a meme, or a decision) propagates through a population of individuals. The mechanism underlying the process is hidden. However, the mechanism is of outstanding interest in all three cases, since understanding diffusion is necessary for stopping infections, predicting meme propagation, or maximizing sales of a product. 

This article presents a method for inferring the mechanisms underlying diffusion processes based on observed infections. To achieve this aim, we construct a model incorporating some basic assumptions about the spatiotemporal structures that generate diffusion processes. The assumptions are as follows. First, diffusion processes occur over static (fixed) but unknown networks (directed graphs). Second, infections are binary, \ie, a node is either infected or it is not; we do not model partial infections or the partial propagation of information. Third, infections along edges of the network occur independently of each other. Fourth, an infection can occur at different times: the likelihood of node $a$ infecting node $b$ at time $t$ is modeled via a probability density function depending on $a$, $b$ and $t$. Finally, we observe \emph{all} infections occurring in the network during the recorded time window. Our aim is to infer the connectivity of the network and the likelihood of infections across its edges after observing the times at which nodes in the network become infected.




In more detail, we formulate a generative probabilistic model of diffusion that aims to describe realistically how infections occur over time in a static network. Finding the optimal network and transmission rates maximizing the likelihood of an observed set of infection cascades reduces to solving a convex program. The convex problem decouples into many smaller problems, allowing for natural parallelization so that our algorithm scales to networks with hundreds of thousands of nodes. We show the effectiveness of our method by reconstructing the connectivity and continuous temporal dynamics of synthetic and real networks using cascade data. 

\xhdr{Related work}
The work most closely related to ours~\citep{manuel10netinf, meyers10netinf} also uses a generative probabilistic model for inferring diffusion networks. \citet{manuel10netinf} (\netinf) infers network connectivity using submodular optimization and \citet{meyers10netinf} (\connie) infer not only the connectivity but also a prior probability of infection for every edge using a convex program and some heuristics. 
However, both papers force the transmission rate between all nodes to be fixed -- and not inferred. In contrast, our model allows transmission at different rates across different edges so that we can infer temporally heterogeneous interactions within a network, as found in real-world examples. Thus, we can now infer the temporal dynamics of the underlying network.

The main innovation of this paper is to model diffusion as a spatially discrete network of continuous, conditionally independent temporal processes occurring at different rates. 
Infection transmission depends on the complex intricacies of the underlying mechanisms (\eg, a person's susceptibility to viral infections depends on weather, diet, age, stress levels, prior exposures to similar pathogens and so on).  We avoid modeling the mechanisms underlying individual infections, and instead develop a data-driven approach, suitable for large-scale analyses, that infers the diffusion process using only the visible spatiotemporal traces (cascades) it generates. We therefore model diffusion using only time-dependent pairwise transmission likelihood between pairs of nodes, transmission rates and infection times, but not prior probabilities of infection that depend on unknown external factors. 
To the best of our knowledge, continuous temporal dynamics of diffusion networks has not been modeled or inferred in previous work. We believe this is a key point for understanding diffusion processes.
%

\section{Problem formulation}
\label{sec:formulation}
This paper develops a method for inferring the spatiotemporal dynamics that generate observed infections. In this section we formulate our model, starting from the data it is designed for, and concluding with a precise statement of the network inference problem.
\begin{table*}[t]
    \caption{}
    \label{tab:models}
    \vskip 0.1in
    \begin{center}
    \begin{small}
    \begin{tabular*}{\textwidth}{@{\extracolsep{\fill}} l  c l c c}
	\toprule
        \multirow{2}{*}{\textbf{Model}}
		& \multicolumn{2}{c}{\textbf{Transmission likelihood}} & \textbf{Log survival function} & \textbf{Hazard function} \vspace{1mm} \\ 
		& \multicolumn{2}{c}{$f(t_i|t_j;\alpha_{j,i})$} & $\log S(t_i|t_j;\alpha_{j,i})$ & $H(t_i|t_j;\alpha_{j,i})$ \\
	\midrule
	Exponential (\expo) & $\left\{
	\begin{array}{l l}
    	\alpha_{j,i}\cdot e^{-\alpha_{j,i} (t_i-t_j)} \\
 		0
	\end{array}\right.$ &
	\hspace{-8mm}
	$\begin{array}{l}
		\text{if $t_j < t_i$}\\
		\text{otherwise}
	\end{array}$
	& $-\alpha_{j,i} (t_i-t_j)$ & $\alpha_{j,i}$ \\
	\midrule
	Power law (\pow) & $\left\{
	\begin{array}{l}
    	\frac{\alpha_{j,i}}{\delta} \left(\frac{t_i-t_j}{\delta}\right)^{-1-\alpha_{j,i}}\\
		0
	\end{array}\right.$ &
	\hspace{-8mm}
	$\begin{array}{l}
		\text{if $t_j+\delta < t_i$}\\
		\text{otherwise}\\
	\end{array}$
	& $-\alpha_{j,i} \log\left(\frac{t_i-t_j}{\delta}\right)$ & $\alpha_{j,i}\cdot\frac{1}{t_i-t_j}$ \\
	\midrule
	Rayleigh (\ray) & $\left\{
	\begin{array}{l}
    	\alpha_{j,i} (t_i-t_j) e^{-\frac{1}{2} \alpha_{j,i} (t_i-t_j)^2}\\
		0
	\end{array}\right.$ &
	\hspace{-8mm}
	$\begin{array}{l}
		\text{if $t_j < t_i$}\\
		\text{otherwise}
	\end{array}$
	& $-\alpha_{j,i} \frac{(t_i-t_j)^2}{2}$ & $\alpha_{j,i} \cdot (t_i-t_j)$\\
	\bottomrule
    \end{tabular*}
    \end{small}
    \end{center}
    \vskip -0.1in
\end{table*}

\xhdr{Data}
Observations are recorded on a fixed population of $N$ nodes and consist of a set $C$ of cascades $\{\casc^1,\ldots,\casc^{|C|}\}$. Each cascade $\casc^c$ is a record of observed infection times within the population during a time interval of length $T^c$. A cascade is an $N$-dimensional vector $\casc^c:=(t^c_1,\ldots,t^c_N)$ recording when nodes are infected, $t^c_k\in [0,T^c]\cup\{\unobs\}$. Symbol $\unobs$ labels nodes that are not infected during observation window $[0,T^c]$ -- it does not imply that nodes are never infected. The `clock' is reset to 0 at the start of each cascade. Lengthening the observation window $T^c$ increases the number of observed infections within a cascade $c$ and results in a more representative sample of the underlying dynamics. However, these advantages must be weighed against the cost of observing for longer periods. For simplicity we assume $T^c = T$ for all cascades; the results generalize trivially.

The time-stamps assigned to nodes by a cascade induce the structure of a directed acyclic graph (DAG) on the network (which is \emph{not} acyclic in general) by defining
node $i$ is a parent of $j$ if $t_i < t_j$. Thus, it is meaningful to refer to parents and children within a cascade, but not on the network. The DAG structure dramatically 
simplifies the computational complexity of the inference problem. Also, since the underlying network is inferred from many cascades (each of which imposes its 
own DAG structure), the inferred network is typically not a DAG.

\xhdr{Pairwise transmission likelihood}
The first step in modeling diffusion dynamics is to consider pairwise interactions. We assume that infections can occur at different rates over different edges of a network, and aim to infer the transmission rates between pairs of nodes in the network.

Define $f(t_i | t_j, \alpha_{j,i})$ as the conditional likelihood of transmission between a node $j$ and node $i$.  The transmission likelihood depends on the infection 
times $(t_j, t_i)$ and a pairwise transmission rate $\alpha_{j,i}$. A node cannot be infected by a node infected later in time. In other words, a node $j$ that has been
infected at a time $t_j$ may infect a node $i$ at a time $t_i$ only if $t_j < t_i$. Although in some scenarios it may be possible to estimate a non-parametric likelihood
empirically, for simplicity we consider three well-known parametric models: exponential, power-law and Rayleigh (see Table~\ref{tab:models}). Transmission rates are denoted as $\alpha_{j,i}\geq0$ and $\delta$ is the minimum allowed time difference in the power-law to have a bounded likelihood. As $\alpha_{j,i} \rightarrow 0$ the likelihood of infection tends to zero and the expected transmission time becomes arbitrarily long. Without loss of generality, we consider $\delta = 1$ in the power-law model from now on.

Exponential and power-laws are monotonic models that have been previously used in modeling diffusion networks and social networks~\citep{manuel10netinf, meyers10netinf}. 
Power-laws model infections with long-tails. The Rayleigh model is a non-monotonic parametric model previously used in epidemiology~\citep{wallinga04epidemic}. It is well-adapted to modeling fads, where infection likelihood rises to a peak and then drops extremely rapidly.

We recall some additional standard notation~\cite{lawless1982statistical}. The cumulative density function, denoted $F(t_i | t_j ; \alpha_{j,i})$, is computed from the transmission likelihoods. Given that node $j$ was infected at time $t_j$, the \emph{survival function} of edge $j\rightarrow i$ is the probability that node $i$ is \emph{not} infected by node $j$ by time $t_i$:
\begin{equation*}
	S(t_i | t_j ; \alpha_{j,i}) = 1-F(t_i | t_j ; \alpha_{j,i}).
\end{equation*}
The \emph{hazard function}, or instantaneous infection rate, of edge $j\rightarrow i$ is the ratio
\begin{equation*}
	H(t_i | t_j ; \alpha_{j,i}) = \frac{f(t_i | t_j ; \alpha_{j,i})}{S(t_i | t_j ; \alpha_{j,i})}.
\end{equation*}
The hazard functions of our models are simple, Table~\ref{tab:models}.

\xhdr{Probability of survival given a cascade}
We compute the probability that a node survives uninfected until time $T$, given that some of its parents are already infected. 
Consider a cascade $\casc:=(t_1, \ldots, t_N)$ and a node $i$ \emph{not} infected during the observation window, $t_i >T$. Since each 
infected node $k$ may infect $i$ independently, the probability that nodes $1\ldots N$ do \emph{not} infect node 
$i$ by time $T$ is the product of the survival functions of the infected nodes $1 \ldots N | t_k \leq T$ targeting $i$,
\begin{equation} \label{eq:prob1_F}
	\prod_{t_k \leq T} S(T | t_k ; \alpha_{k,i}).
\end{equation}

\xhdr{Likelihood of a cascade}
Consider a cascade $\casc:=(t_1, \ldots, t_N)$. We first compute the likelihood of the observed infections $\casc^{\leq T}=(t_1,\ldots,t_N|t_i\leq T)$. 
Since we assume infections are conditionally independent given the parents of the infected nodes, the likelihood factorizes over nodes as
\begin{equation}
	\label{e:factor}
	f(\casc^{\leq T}; \alphs)=\prod_{t_i\leq T} f(t_i|t_1,\ldots,t_N\setminus t_i ; \alphs).
\end{equation}
Computing the likelihood of a cascade thus reduces to computing the conditional likelihood of the infection time of each node given the rest of the cascade. As 
in the independent cascade model~\citep{kempe03maximizing}, we assume that a node gets infected once the \emph{first} parent infects the node. Given 
an infected node $i$, we compute the likelihood of a potential parent $j$ to be the first parent by applying Eq.~\ref{eq:prob1_F},
\begin{equation}
	f(t_i | t_j ; \alpha_{j,i}) \times \prod_{j \neq k, t_k < t_i} S(t_i | t_k ; \alpha_{k,i}).
\end{equation}
We now compute the conditional likelihoods of Eq.~\ref{e:factor} by summing over the likelihoods of the mutually disjoint events that each potential 
parent is the first parent,
\begin{multline}
	f(t_i | t_1, \ldots, t_N \setminus t_i ; \alphs) = \sum_{j : t_j < t_i} f(t_i | t_j ; \alpha_{j,i}) \times \\
	\prod_{j \neq k, t_k < t_i} S(t_i | t_k ; \alpha_{k,i}).
\end{multline}
By Eq.~\ref{e:factor} the likelihood of the infections in a cascade is
\begin{multline}
	f(\casc^{\leq T} ; \alphs) = \prod_{t_i\leq T} \sum_{j : t_j < t_i}
	f(t_i | t_j ; \alpha_{j,i}) \times \\
	\prod_{k : t_k < t_i, k \neq j} S(t_i | t_k ; \alpha_{k,i}).
\end{multline}
Removing the condition $k \neq j$ makes the product independent of $j$,
\begin{multline} \label{eq:casc}
	f(\casc^{\leq T}; \alphs) = \prod_{t_i\leq T} \prod_{k : t_k < t_i}
	S(t_i | t_k ; \alpha_{k,i}) \times \\
	\sum_{j : t_j < t_i} \frac{f(t_i | t_j ; \alpha_{j,i})}{S(t_i | t_j ; \alpha_{j,i})}.
\end{multline}
Eq.~\ref{eq:casc} only considers infected nodes. However, the fact that some nodes are \emph{not} infected during the observation window is also informative. We therefore add the multiplicative survival term from Eq.~\ref{eq:prob1_F} and also replace the ratios in Eq.~\ref{eq:casc} with hazard functions:
\begin{multline} \label{eq:loglikelihood}
	f(\casc;\alphs) = \prod_{t_i \leq T} \prod_{t_m > T} S(T | t_i ; \alpha_{i,m}) \times \\
	\prod_{k : t_k < t_i} S(t_i | t_k ; \alpha_{k,i}) \sum_{j : t_j < t_i}
	H(t_i | t_j ; \alpha_{j,i}).
\end{multline}
Assuming independent cascades, the likelihood of a set of cascades $C = \{\casc^1,\ldots,\casc^{|C|}\}$ is the product of the likelihoods of the individual cascades given by Eq.~\ref{eq:loglikelihood}:
\begin{align}
	\label{eq:cloglikelihood}
	\prod_{\casc^c \in C} f(\casc^c ; \alphs).
\end{align}

\xhdr{Network Inference Problem}
Our goal is to find the transmission rates $\alpha_{j,i}$ of every pair of nodes such that the likelihood of an observed set of cascades $C = \{\casc^1,\ldots,\casc^{|C|}\}$ 
is maximized. Thus, we aim to solve:
\begin{equation} 
	\label{eq:optproblem}
	\begin{array}{ll}
		\mbox{minimize$_{\alphs}$} & - \sum_{c \in C} \log f(\casc^c;\alphs) \\ 
		\mbox{subject to} & \alpha_{j,i} \geq 0,\, i,j=1,\ldots,N, i \neq j,
	\end{array}
\end{equation}
where $\alphs := \{\alpha_{j,i}\,|\, i,j=1,\ldots,n, i \neq j\}$ are the variables. The edges of the network are 
those pairs of nodes with transmission rates $\alpha_{j,i} > 0$.

\section{Proposed algorithm: \netrate}
\label{sec:proposed}
The solution to Eq.~\ref{eq:optproblem} is unique, computable and consistent:

\begin{theorem}
	Given log-concave survival functions and concave hazard functions in the parameter(s) of the pairwise transmission likelihoods, the network inference problem defined by equation Eq.~\ref{eq:optproblem} is convex in $\alphs$.
\end{theorem}
\begin{proof}	
	 By Eq.~\ref{eq:cloglikelihood}, the log-likelihood of a set of cascades is
	\begin{multline} \label{eq:lcasc}
		L\big(\{\casc^1\ldots\casc^{|C|}\}; \alphs\big) = \\
		\sum_c \Psi_1(\casc^c; \alphs) + \Psi_2(\casc^c; \alphs) + \Psi_3(\casc^c; \alphs),
	\end{multline}
	where for each cascade $\casc^c \in \{\casc^1,\ldots,\casc^{|C|}\}$,
	\begin{align*}
		\Psi_1(\casc^c; \alphs) &= \sum_{i : t_i \leq T} \sum_{t_m > T}
		\log S(T | t_i ; \alpha_{i,m}),\\
		\Psi_2(\casc^c; \alphs) & = \sum_{i : t_i \leq T} \sum_{j : t_j < t_i}
		\log S(t_i | t_j ; \alpha_{j,i}), \\
		\Psi_3(\casc^c; \alphs) &= \sum_{i : t_i \leq T} \log \left(\sum_{j : t_j < t_i}
		H(t_i | t_j ; \alpha_{j,i})\right).
	\end{align*}
If all pairwise transmission likelihoods between pairs of nodes in the network have log-concave survival functions and concave hazard functions in the parameter(s) of the pairwise transmission likelihoods, then convexity of Eq.~\ref{eq:optproblem} follows from linearity, composition rules for concavity, and concavity of the logarithm. 
\end{proof}

\begin{corollary}
	The network inference problem defined by equation Eq.~\ref{eq:optproblem} is convex for the exponential, power-law and Rayleigh models.
\end{corollary}
\vspace{.5mm}
\begin{theorem}
	The maximum likelihood estimator $\hat{\alpha}$ given by the solution of Eq.~\ref{eq:optproblem} is consistent.
\end{theorem}
\begin{xsketch}
	We check the criteria for consistency of identification, continuity and compactness \cite{newey94}.
	The log-likelihood in Eq.~\ref{eq:lcasc} is a continuous function of $\alphs$ for any fixed set of cascades $\{\casc^1\ldots\casc^{|C|}\}$, and each $\alpha$ defines a unique function $\log f(\cdot|\alphs)$ on the set of cascades. Finally, note that $L\rightarrow -\infty$ for both $\alpha_{ij}\rightarrow 0$ and $\alpha_{ij}\rightarrow \infty$ for all $i,j$ so we lose nothing imposing upper and lower bounds thus restricting to a compact subset.
	
\end{xsketch}

We refer to our network inference method as \netrate.
\begin{figure*}[!!t]
\centering
  \subfigure[Precision-recall (Hierarchical, \expo)]{\includegraphics[width=0.28\textwidth]{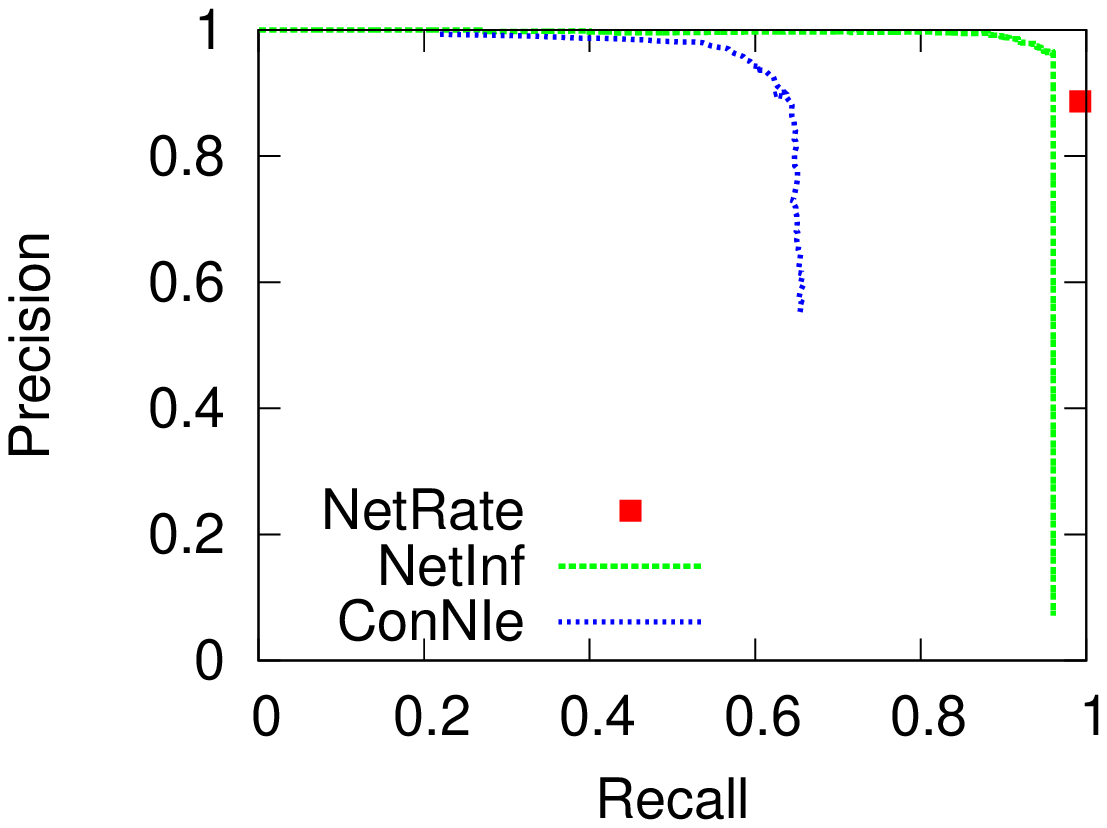} \label{fig:pr_exp_hie}} 
  \subfigure[Precision-recall (Forest fire, \pow)]{\includegraphics[width=0.28\textwidth]{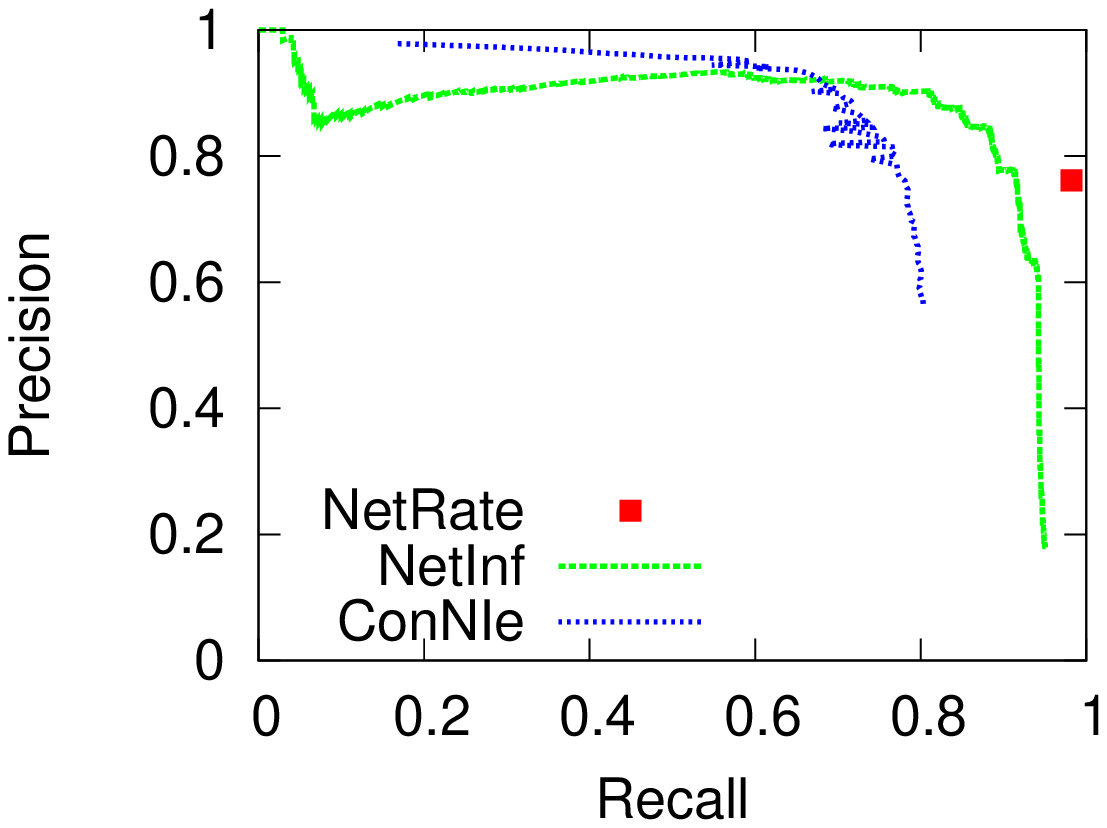} \label{fig:pr_pl_ff}} 
  \subfigure[Precision-recall (Random, \ray)]{\includegraphics[width=0.28\textwidth]{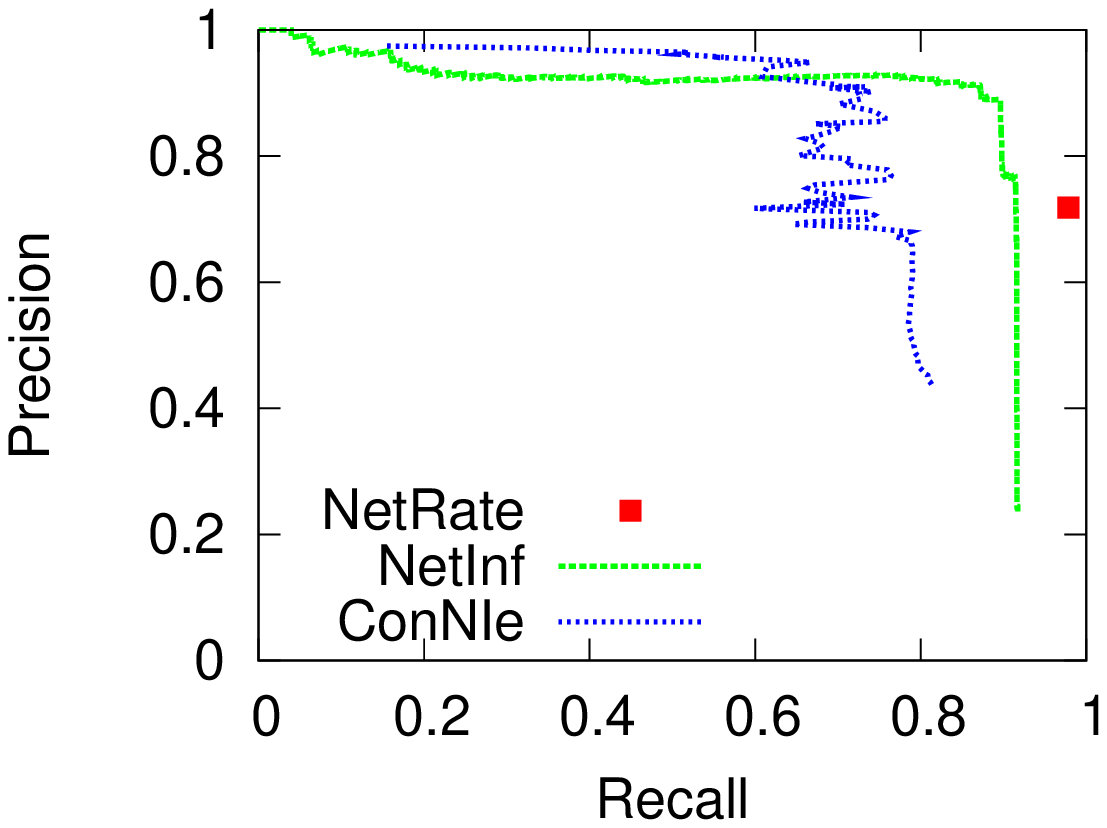} \label{fig:pr_rayleigh_random}} \\
  \subfigure[Accuracy (Hierarchical, \expo)]{\includegraphics[width=0.28\textwidth]{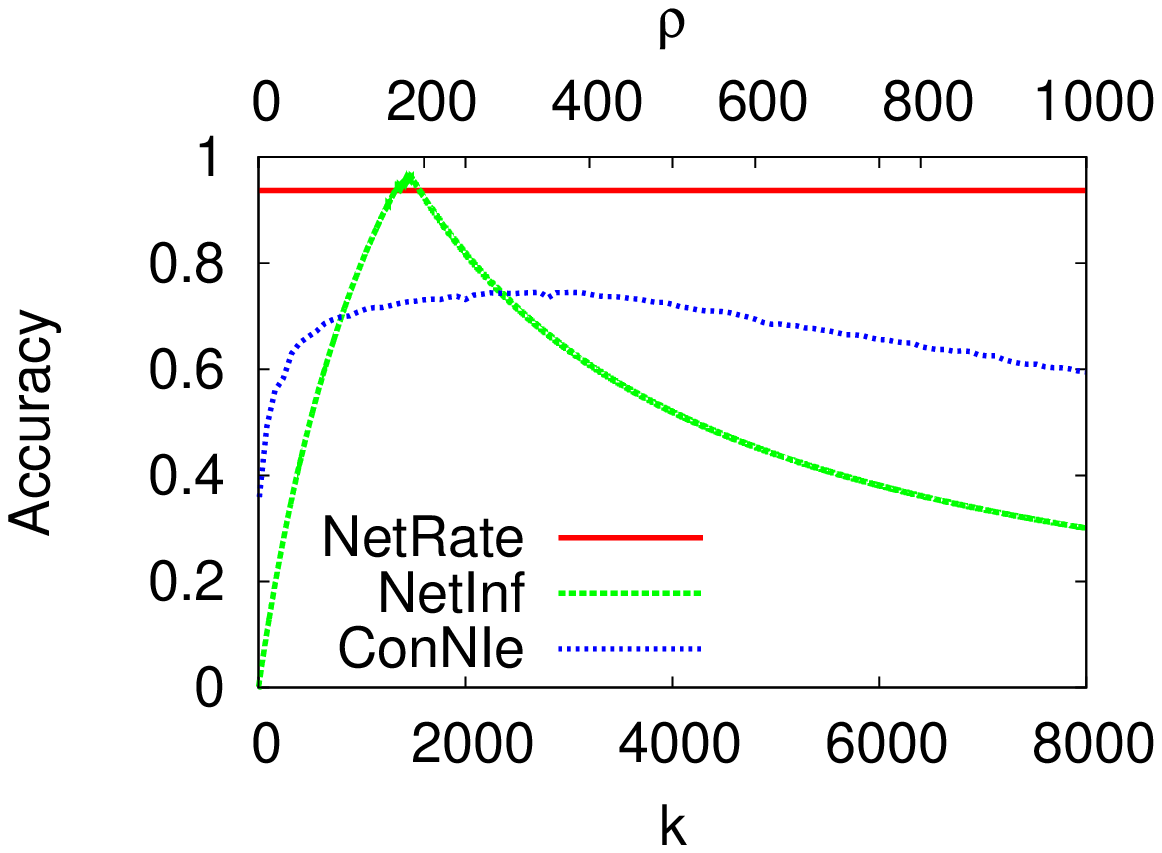} \label{fig:acc_exp_hie}} 
  \subfigure[Accuracy (Forest fire, \pow)]{\includegraphics[width=0.28\textwidth]{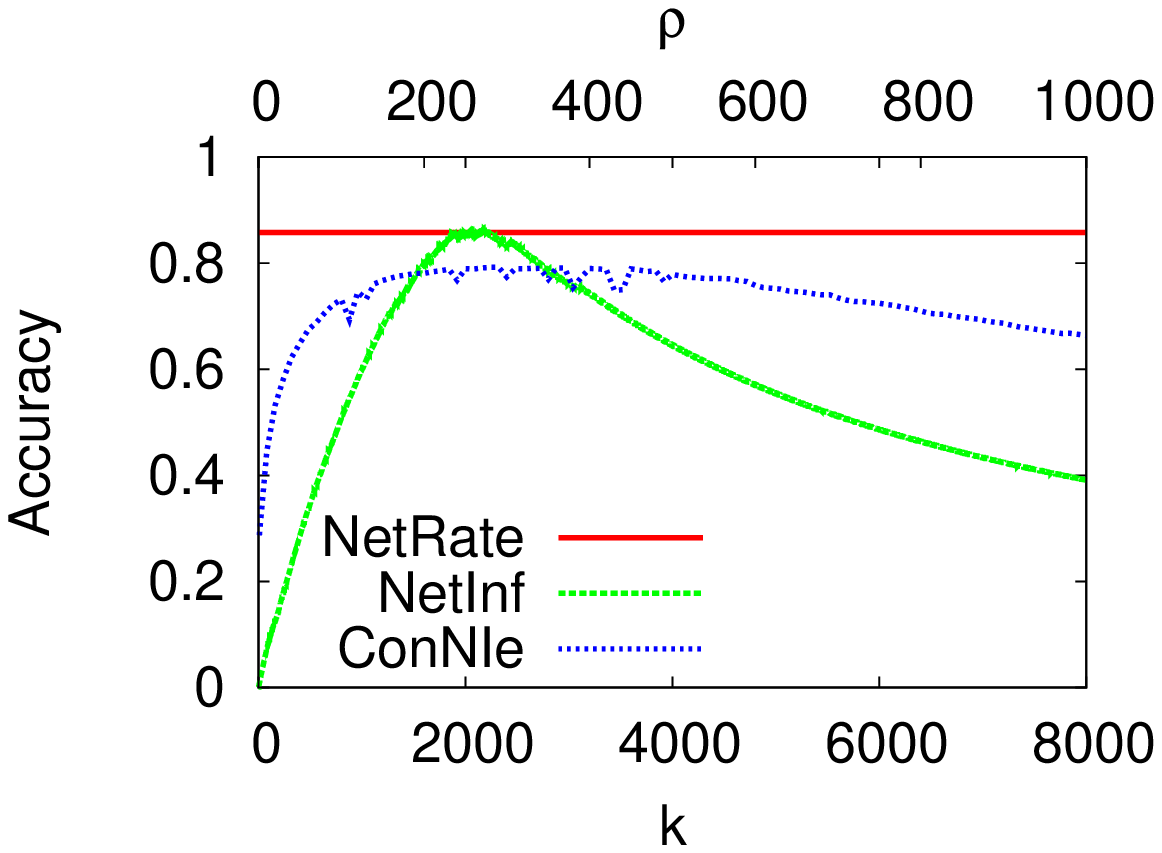} \label{fig:acc_pl_ff_5000}}
  \subfigure[Accuracy (Random, \ray)]{\includegraphics[width=0.28\textwidth]{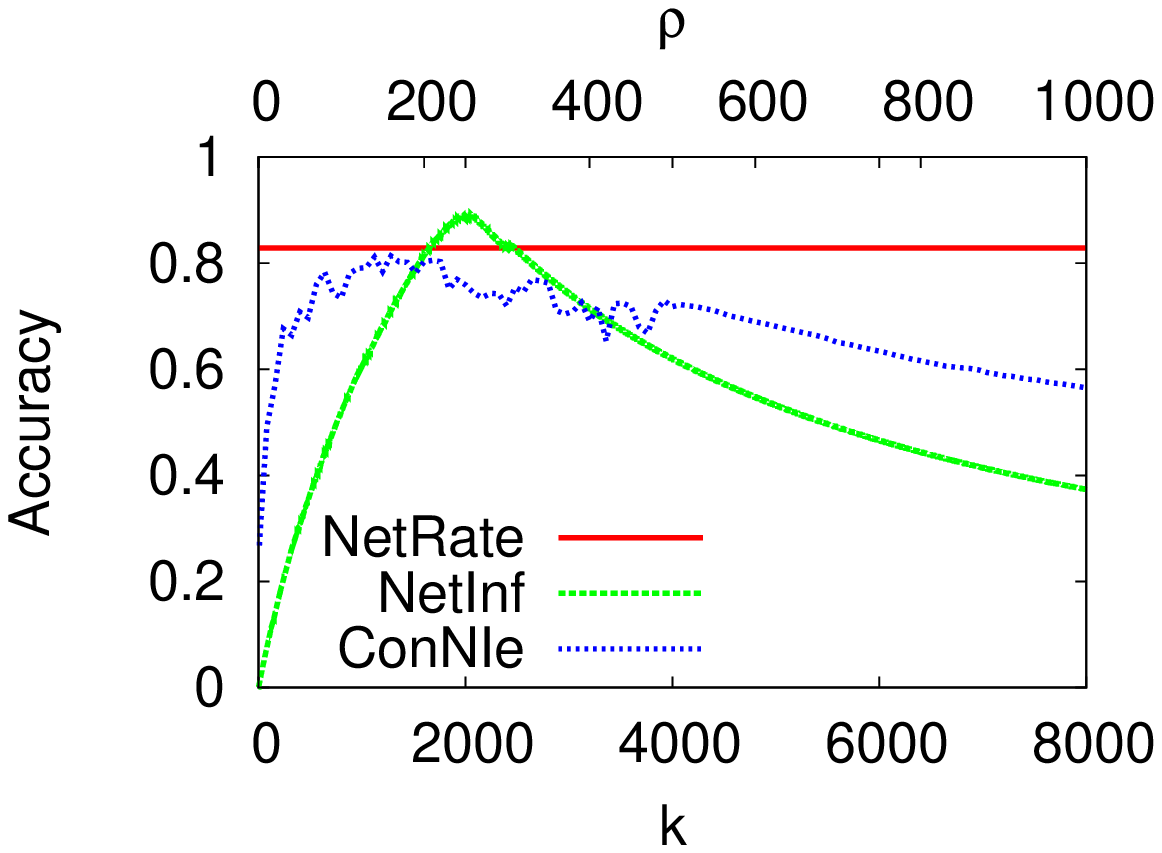} \label{fig:acc_rayleigh_random}}
	\caption{
	Panels (a-c) plot precision against recall; panels (d-f) plot accuracy. For \connie and \netinf we sweep over parameters $\rho$ (penalty factor) and $k$ (number of edges) respectively to control the solution sparsity in both algorithms, thereby generating a family of inferred models. \netrate has no tunable parameters and therefore yields a unique solution.
	(a,d): 1,024 node hierarchical Kronecker network with exponential model for 5,000 cascades.
	(b,e): 1,024 node Forest Fire network with power law model for 5,000 cascades.
	(c,f): 1,024 node random Kronecker network with Rayleigh model for 2,000 cascades. 	
	}
	\label{fig:pr}
\end{figure*}

\xhdr{Properties of \netrate}
We highlight some common features of the solutions to the network inference problem for the exponential, power-law and Rayleigh models. 

All terms in Eq.~\ref{eq:lcasc} depend only on transmission rates $\alpha_{j, i}$ and infection time differences $(t_i - t_j)$, but not absolute infection times $t_i$ or $t_j$. 
Our formulation thus does not depend on the absolute time of the root node of each cascade.

The $\Psi_1$ and $\Psi_2$ terms contribute a positively weighted $l_1$-norm on vector $\alphs$ that encourages sparse solutions~\citep{boyd2004convex}. The penalty arises 
naturally within the probabilistic model and therefore heuristic penalty terms to encourage sparsity are not necessary. Each term of the $l_1$-norm is 
linearly (exponential model), logarithmically (power-law) or quadratically (Rayleigh) weighted by infection times.

The $\Psi_2$ term penalizes edges $k \rightarrow i$ based on the infection times difference $t_i-t_k$. Edges transmitting infections slowly are heavily penalized and conversely. The $\Psi_1$ term 
penalizes edges $i \rightarrow j$ targeting \emph{uninfected} nodes $j$ based on the time $T-t_i$ till the observation window cutoff. Lengthening the observation window produces harsher penalties -- 
however, it also allows further infections. The penalties are finite, \ie, if no infection of node $j$ is observed, we can only say that it has survived until time $T$. There is 
insufficient evidence to claim $j$ will never be infected. \netrate does not use empirically ungrounded parameters (such as number of edges $k$ and penalty factor $\rho$ used by \netinf and \connie respectively) to leap from not observing an infection to inferring it is impossible. Instead, \netrate infers that infections are impossible across certain edges, \ie, that some of the optimal rates $\alpha_{j, i}$ are $0$, based solely on the observed data and the length of the time horizon.

The $\Psi_3$ term ensures infected nodes have at least one parent since otherwise the objective function would be negatively unbounded, \ie, $\log 0=-\infty$. Moreover, our formulation encourages a natural diminishing property on the number of parents of a node -- since the logarithm grows \emph{slowly}, it weakly rewards infected nodes for having many parents.

\xhdr{Optimizing \netrate}
We speed up the convex program by orders of magnitude via two improvements: 

{\em Distributed optimization:} The optimization problem splits into $N$ subproblems, one for each node $i$, in which we find $N-1$ rates $\alpha_{j, i},\, j=1,\ldots,N \setminus i$. The computation can be performed in parallel, obtaining local solutions that are globally optimal. Importantly, each node's computation only requires the infection times of other nodes in cascades it belongs to. 

{\em Unfeasible rates:} If a pair $(j, i)$ is not in any common cascades, $\alpha_{j,i}$ only arises in the non-positive term $\Psi_3$ in Eq.~\ref{eq:lcasc}, so the optimal 
$\alpha_{j, i}$ is zero. We therefore simply the objective function by setting $\alpha_{j,i}$ to zero.

\xhdr{Solving \netrate}
We solve Eq.~\ref{eq:optproblem} with \texttt{CVX}, a package for specifying and solving convex programs~\citep{cvx}. 

\section{Experimental evaluation}
\label{sec:evaluation}
We evaluate the performance of \netrate on: (i) synthetic networks that mimic the structure of social networks and (ii) real cascades extracted from the MemeTracker dataset\footnote{Data available at \url{http://memetracker.org}}. 
We show that \netrate discovers more than 95\% of the edges in synthetic networks and more than 60\% in real networks, accurately recovers transmission rates from diffusion data, and typically outperforms two previously developed inference algorithms, \netinf and \connie. We use the public implementations of \netinf and \connie.
\begin{figure}[t]
\centering
	\includegraphics[width=0.28\textwidth]{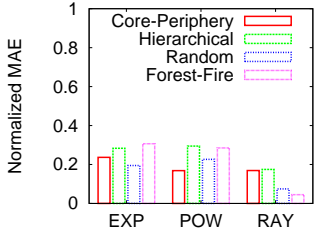}
 	\caption{\netrate'{}s normalized mean absolute error (MAE) for three types of Kronecker networks 
(1,024 nodes and 2,048 edges) and a Forest Fire network (1,024 edges and 2,422 edges) for 5,000 cascades. We consider all three 
models of transmission: exponential (\expo), power-law (\pow) and Rayleigh (\ray).}
	\label{fig:mae}
\end{figure}

\subsection{Experiments on synthetic data}

\xhdr{Experimental setup}
We focus on synthetic networks that mimic the structure of real-world diffusion networks -- in particular, social networks. We consider two models of directed real-world social networks: the Forest Fire (scale free) model~\citep{barabasi99emergence} and the Kronecker Graph model~\citep{leskovec2010kronecker} to generate diffusion networks. We generate three types of Kronecker graph with very different structures: random~\cite{erdos60random} (parameter matrix $[0.5, 0.5; 0.5, 0.5]$), hierarchical~\citep{clauset08hierarchical} ($[0.9, 0.1; 0.1, 0.9]$) and core-periphery~\citep{jure08ncp} ($[0.9, 0.5; 0.5, 0.3]$). 

First, we generate network $G^*$ by drawing transmission rates for edges $(j, i)$ from a uniform distribution. For the exponential and Rayleigh models $\alpha\in [0.01, 1]$ and for the power law $\alpha\in [0.01,2]$. The transmission rate for an edge $(j,i)$ models how fast the information spreads from node $j$ to node $i$ in social networks. Then, we generate a set of cascades over $G^*$. Root nodes of cascades are chosen uniformly at random. As noted previously, the optimization problem depends on the time 
differences $(t_i - t_j)$. Therefore, our formulation does not depend on the absolute time of the root node of each cascade. Once 
a node is infected, the transmission likelihoods of outgoing edges determine the infection times of its neighbours. We record the time of the first infection if a node is infected more than once. Infections are not observed after a pre-specified time horizon $T$.

\xhdr{Accuracy of \netrate}
We evaluate \netrate against two other inference methods, \netinf and \connie, by comparing the inferred and true networks via three measures: precision, recall and accuracy. Precision is the fraction of edges in the inferred network $\hat{G}$ present in the true network $G^*$ . Recall is the fraction of edges of the true network $G^*$ present in the inferred network $\hat{G}$. Accuracy is $1-\frac{\sum_{i,j} |I(\alpha^*_{i,j})-I(\hat{\alpha}_{i,j})|}{\sum_{i,j}I(\alpha^*_{i,j}) + \sum_{i,j} I(\hat{\alpha}_{i,j})}$, where $I(\alpha)=1$ if $\alpha > 0$ and $I(\alpha)=0$ otherwise. Inferred networks with no edges or only false edges have zero accuracy.
Second, we evaluate how accurately \netrate infers transmission rates over edges by computing the normalized mean absolute error (MAE, \ie, $E\big[|\alpha^*-\hat{\alpha}|/\alpha^*\big]$, where $\alpha^*$ is the true transmission rate and $\hat{\alpha}$ is the estimated transmission rate). Note that in \netrate, as for real cascades, the probability 
of infection depends on both the transmission rate and the observation window. In contrast, \connie assigns probability priors to edges that are defined without reference to 
an observation window. Therefore, the values assigned to edges by \netrate and \connie are not comparable, so we do not compute MAE for \connie.

Figure~\ref{fig:pr} compares the precision, recall and accuracy of \netrate with \netinf and \connie for two types of Kronecker networks (hierarchical community structure and random) and a Forest Fire network over an observation window of length $T = 10$.
In terms of precision-recall, \netrate outperforms \connie and \netinf for all the synthetic examples in the Pareto sense~\citep{boyd2004convex}. More 
specifically, if we set \connie and \netinf'{}s tunable parameters to provide solutions with the same precision as \netrate, \netrate'{}s recall is 
always higher than the other two methods. Strikingly, \connie and \netinf do not achieve \netrate'{}s recall for any precision value.
\netrate outperforms \connie with respect to accuracy for any penalty factor $\rho$ in all the synthetic examples. It is also more accurate than \netinf for most values of $k$ (number of edges). Importantly, \netinf and \connie yield a curve of solutions from which have to select a point blindly (or at best heuristically), whereas \netrate yields a unique solution without any tuning. 

Figure~\ref{fig:mae} shows the normalized MAE of the estimated transmission rates for the same networks, computed on 5,000 cascades. The normalized MAE is under 25\% for almost all networks and transmission models -- surprisingly low given we are estimating more than 2,000 non-zero real numbers.

\xhdr{\netrate performance vs. cascade coverage}
Observing more cascades leads to higher precision-recall and more accurate estimates of the transmission rates. Fi\-gure~\ref{fig:cc} plots the MAE of inferred networks against the number of observed cascades for a hierarchical Kronecker network with all three transmission models.  Estimating transmission rates is considerably harder 
than simply discovering edges and therefore more cascades are needed for accurate estimates. As many as 5,000 cascades are required to obtain normalized MAE values lower 
than 20\%.

\xhdr{\netrate performance vs. time horizon}
Intuitively, the longer the observation window, the more accurately \netrate is able to infer transmission rates. Fi\-gure~\ref{fig:th} confirms this intuition by showing the MAE of inferred networks for different time horizons $T$ for a hierarchical Kronecker with exponential, power-law and Rayleigh transmission models for 5,000 cascades.

\xhdr{\netrate running time}
Fi\-gure~\ref{fig:sc} plots the average running time to infer rates of all incoming edges to a node against number of nodes in a network (the number of edges is 
twice the number of nodes) on a single CPU. Further improvements can be achieved since \netrate naturally splits into a collection of subproblems, one per node. A cluster 
with 25 CPUs can therefore infer a network with 16,000 nodes (and 32,000 edges) in less than 4 hours.
\begin{figure*}[t]
	\centering
	\subfigure[Cascade coverage]{\includegraphics[width=0.28\textwidth]{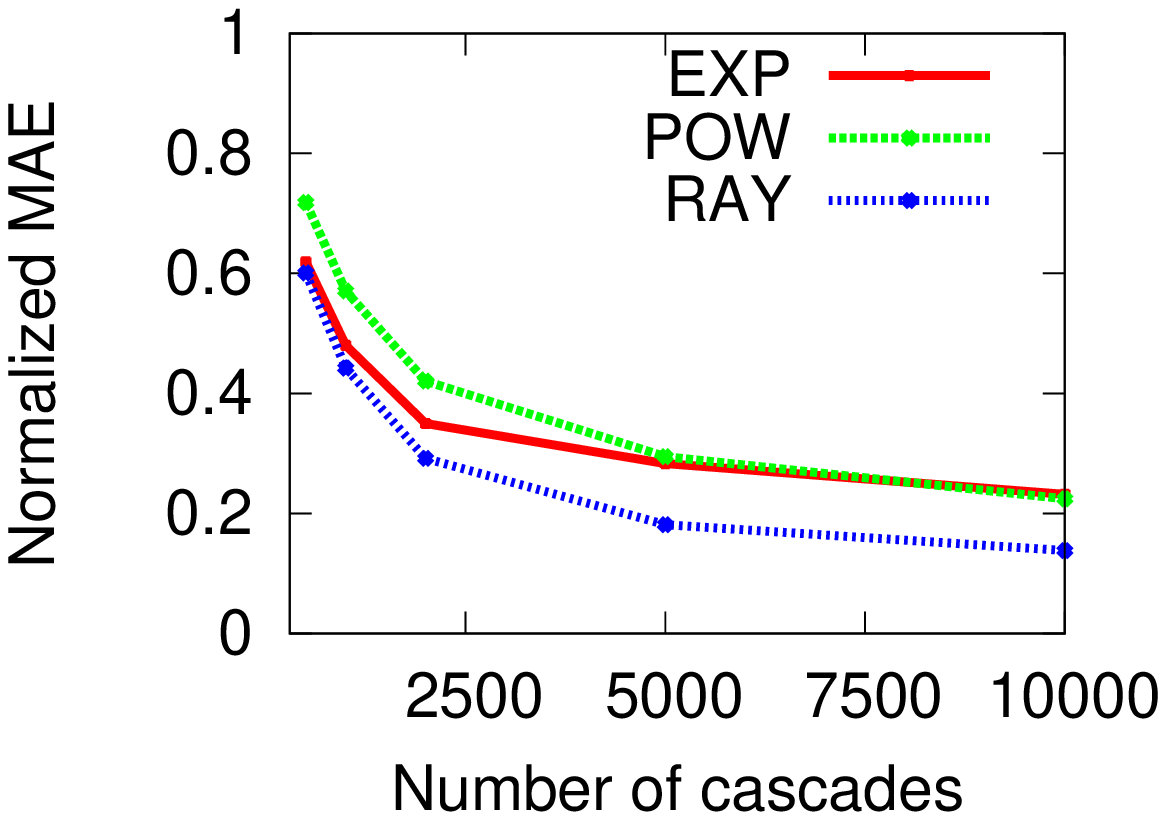} 
	\label{fig:cc}}
	\subfigure[Time horizon]{\includegraphics[width=0.28\textwidth]{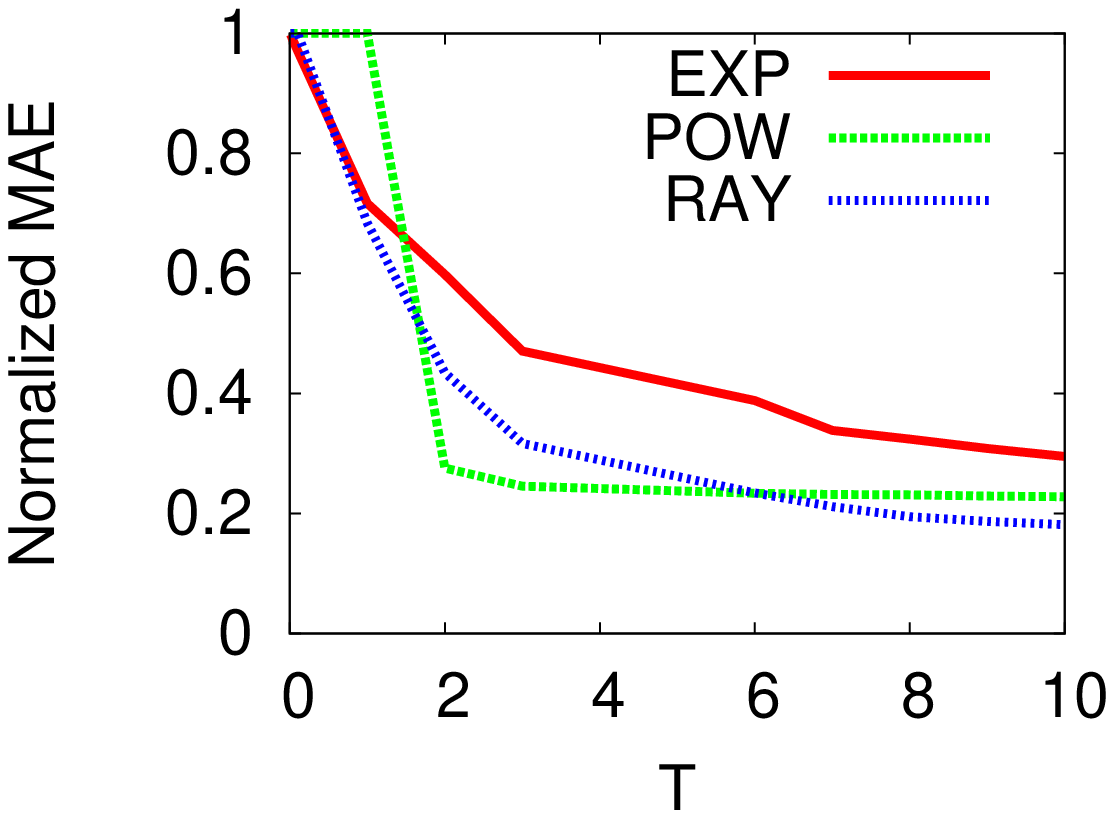}
	\label{fig:th}}
	\subfigure[Running time]{\includegraphics[width=0.28\textwidth]{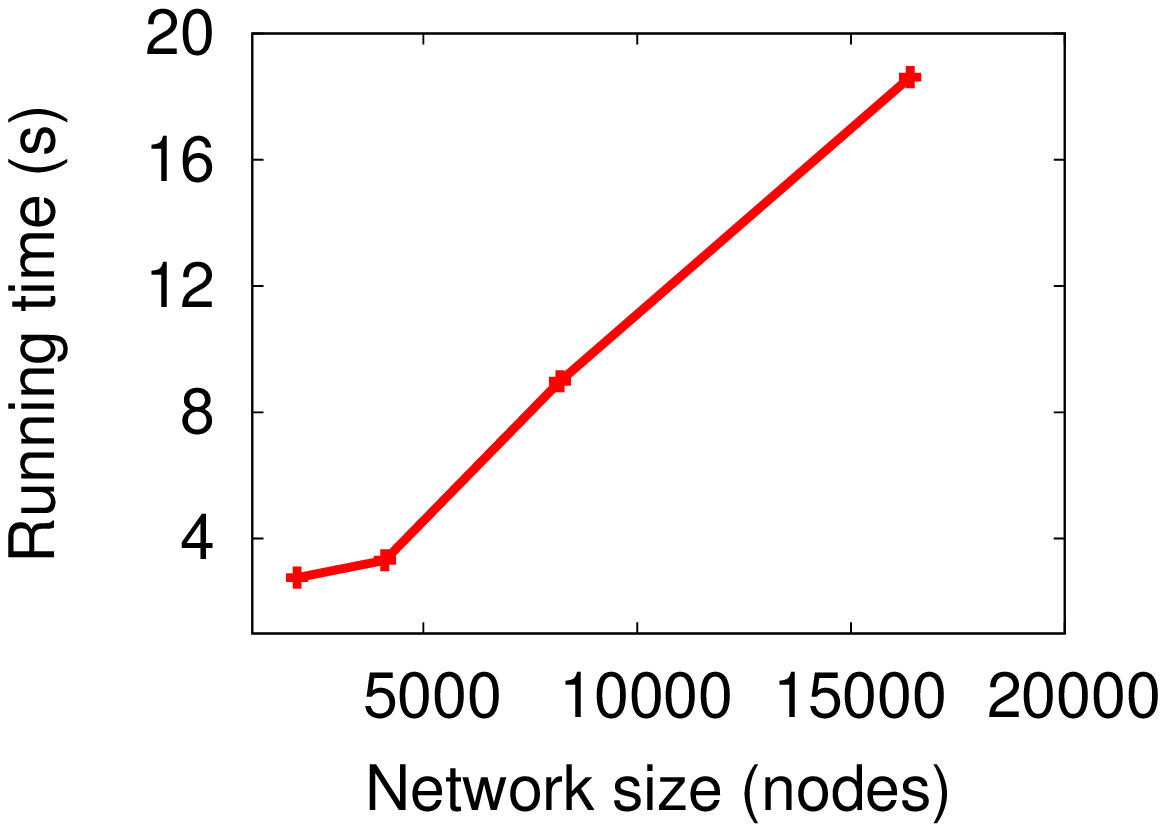}
	\label{fig:sc}}
 	\caption{Panels (a,b) show \netrate'{}s normalized MAE vs number of cascades and time horizon respectively for a hierarchical Kronecker network with 1,024 
nodes and 2,048 edges with exponential (\expo), power-law (\pow) and Rayleigh (\ray) transmission models. Panel (c) plots \netrate'{}s average running time to infer rates of
all incoming edges to a node against network size (number of nodes) for a hierarchical Kronecker network.
	}
\end{figure*}

\subsection{Experiments on real data}

\xhdr{Dataset description}
As in previous work tackling diffusion networks, we use the MemeTracker dataset, which contains more than $172$ million news articles and blog posts from $1$ million online sources.
We use hyperlinks between blog posts to trace the flow of information. A site publishes a piece of information and uses hyper-links to refer to the same or closely 
related pieces of information published by other sites. These other sites link to still others and so on. A cascade is thus a collection of time-stamped hyperlinks 
between sites (in blog posts) that refer to the same or closely related pieces of information. We record one cascade per piece -- or closely related pieces --
of information. We extract the top 500 media sites and blogs with the largest number of documents, 5,000 hyperlinks and 116,234 cascades.

\xhdr{Accuracy on real data}
As the ground truth is unknown on real data, we proceed as follows. We create a network where there is an edge $(u,v)$ if a post on a site $u$ linked to a post on a site $v$. We consider this as the ground truth network $G^*$ and we use the hyperlink cascades to infer the network $\hat{G}$ and evaluate how many edges our method estimates properly. We assume an exponential model.

Figure~\ref{fig:pr-real} compares our results with \netinf and \connie. As in the synthetic experiments, \netrate yields a unique solution whereas the other 
algorithms produce curves of solutions. Panel (a) shows that \netrate performs comparably to \netinf and it outperforms \connie on precision and recall. Panel (b) 
plots accuracy: \netrate's outperforms the other two algorithms on the majority of their outputted solutions, and almost matches their best performances. Since there are no principled methods for choosing single solutions for \netinf and \connie, there is no guarantee that the solutions chosen from the curves will be 
anywhere near the highest achievable value.

\section{Conclusions}
\label{sec:conclusions}
We have developed a flexible model, \netrate, of the spatiotemporal structure underlying diffusion processes. The model makes minimal assumptions about the physical, biological or cognitive mechanisms responsible for diffusion. Instead, it infers transmission rates between nodes of a network by computing the model that maximizes the likelihood of the observed data -- temporal traces left by cascades of infections. Qualitative assumptions about infections (\eg, are they long-tailed?) determine the choice  of parametric model on the edges. An interesting feature of \netrate, to be investigated in future work, is the possibility of mixing exponential, power law, Rayleigh or other models within a single inference algorithm, thus providing tremendous flexibility in fitting real data which may combine long-tailed, faddish and other qualitative behaviors.

Remarkably, introducing continuous temporal dynamics, allowing variable transmission rates across edges, and avoiding further assumptions dramatically simplified the problem compared with previous approaches~\citep{manuel10netinf, meyers10netinf}. The model has parameters with natural interpretations, and it leads to a well-defined convex maximum likelihood problem that can be solved efficiently. Importantly, we do not need to tune parameters by hand to control the sparsity of the inferred network (\ie, number of edges to infer or penalty terms). Heuristic $l_1$-like penalty terms, as the ones used in~\citet{meyers10netinf}, are unnecessary since the probabilistic model naturally imposes sparsity.

We evaluated \netrate on a wide range of synthetic diffusion networks with heterogeneous temporal dynamics which aim to mimic the structure of real-world social and information networks. \netrate provides a unique solution to the network inference problem with high recall, precision and accuracy. A direct comparison with the current state of the art 
is difficult, since these methods include a parameter controlling the sparsity of the inferred network that requires blind tuning. Nevertheless, \netrate is typically better in terms of accuracy than previous methods across the full range of their tunable parameters. In addition, it accurately estimates transmission rates, which other methods cannot estimate at all. The performance of \connie appears significantly worse than reported in~\citet{meyers10netinf}; a possible explanation for the degradation is that in
our work, we consider networks with heterogeneous temporal dynamics. It is surprising how well \netinf performs in comparison with \netrate despite assuming uniform temporal dynamics and priors.
\begin{figure}[t]
\centering
  \subfigure[Precision-recall]{\includegraphics[width=0.23\textwidth]{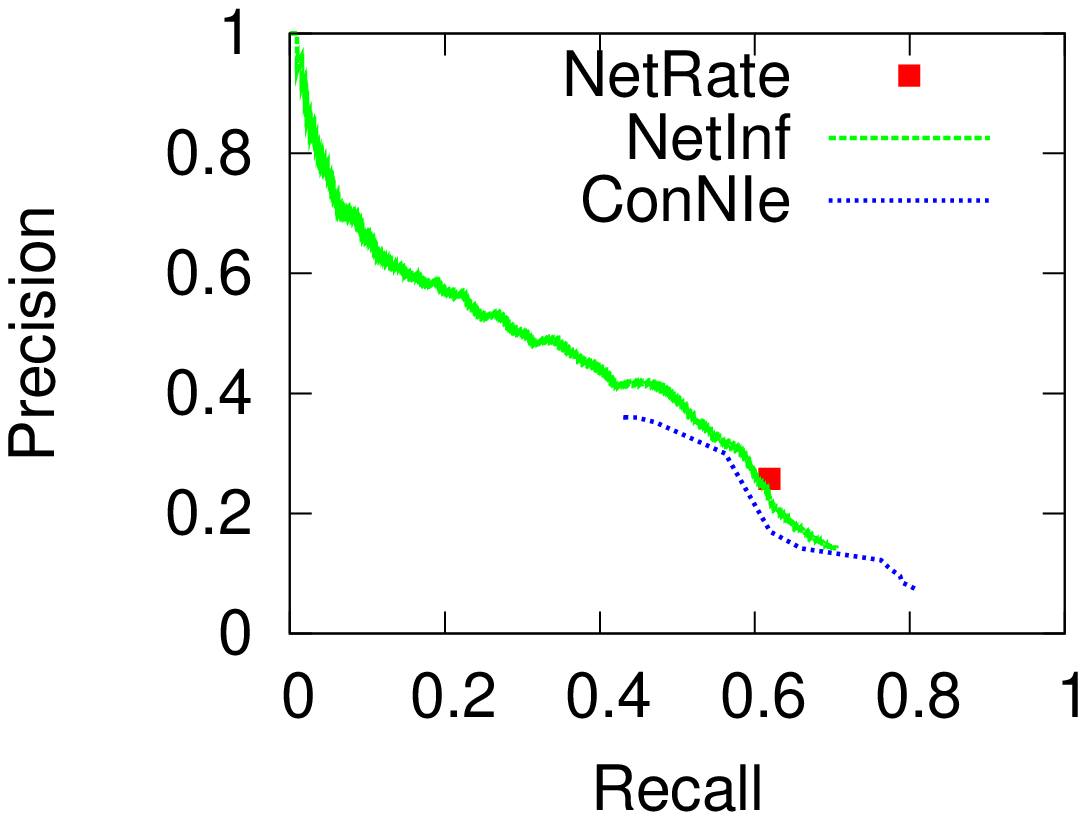}}
  \subfigure[Accuracy]{\includegraphics[width=0.23\textwidth]{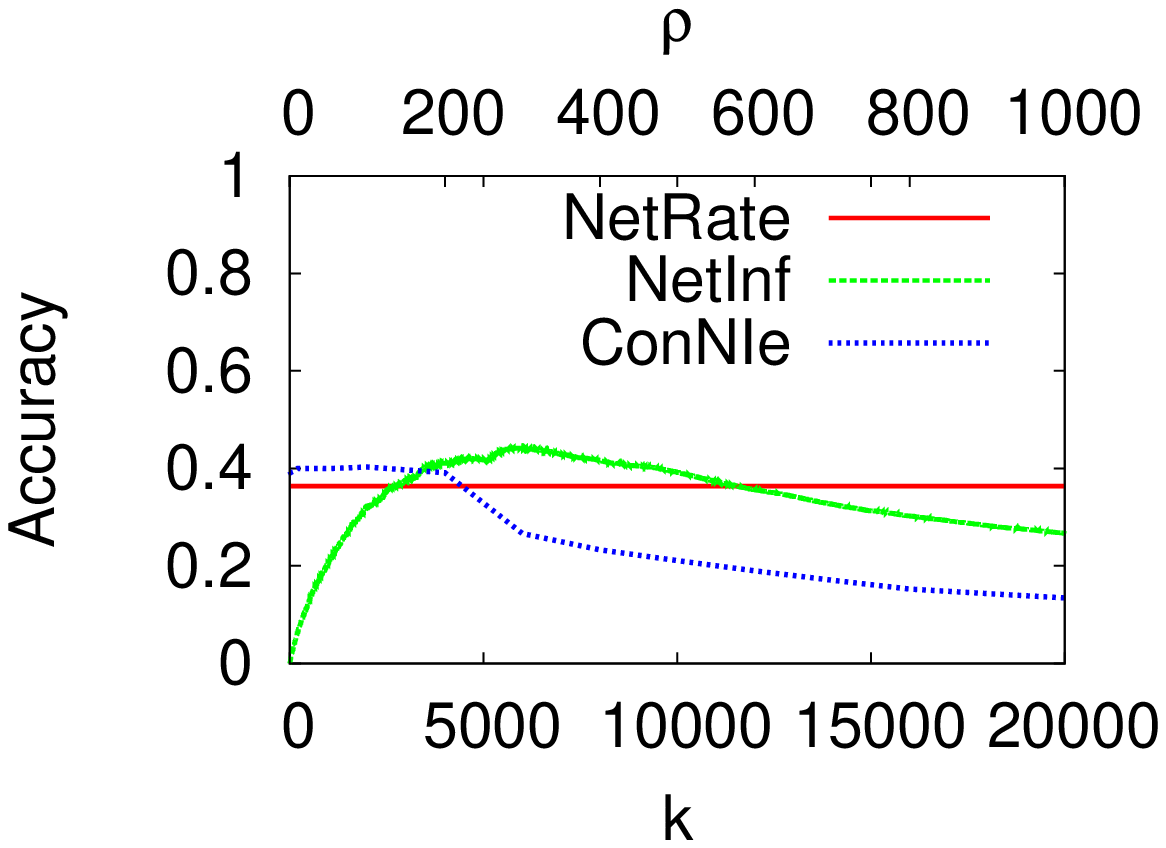}}
  \caption{Real data. Precision-recall and accuracy of \netrate, \netinf and \connie, with an exponential model, on a 500 node hyperlink network with 5,000 edges using hyperlinks cascades.} \label{fig:pr-real}
\end{figure}

Finally, we evaluated \netrate on real data. Again, \netrate provides a unique solution to the network inference problem but in this case, as expected, the values of recall, precision and accuracy are modest -- adopting a simple parametric pairwise transmission model is a simplistic assumption on real data. In terms of accuracy, it outperforms previous methods across a significant part of the full range of their tunable parameters. 

\netrate provides a novel view of diffusion processes. We believe it can be fruitfully applied to several lines of research including influence maximization, control of epidemics, and causal inference.

%
\bibliographystyle{icml2011}
\bibliography{refs}

\end{document}